%% file: main.tex
\documentclass[conference,a4paper]{IEEEtran}
\pdfoutput=1

\input{defs}

\begin{document}

\title{Decoding of Interleaved Reed--Solomon Codes Using Improved Power Decoding}
\author{\IEEEauthorblockN{Sven Puchinger$^{1}$, Johan Rosenkilde n\'e Nielsen$^{2}$}\\
\IEEEauthorblockA{
$^1$Institute of Communications Engineering, Ulm University, Ulm, Germany\\
$^2$Department of Applied Mathematics \& Computer Science, Technical University of Denmark, Lyngby, Denmark\\
Email: \emph{sven.puchinger@uni-ulm.de}, \emph{jsrn@jsrn.dk}
}}

\maketitle

\begin{abstract}
  We propose a new partial decoding algorithm for $m$-interleaved Reed--Solomon (IRS) codes that can decode, with high probability, a random error of relative weight $1-R^{\frac{m}{m+1}}$ at all code rates $R$, in time polynomial in the code length $n$.
  For $m>2$, this is an asymptotic improvement over the previous state-of-the-art for all rates, and the first improvement for $R>1/3$ in the last $20$ years.
  The method combines collaborative decoding of IRS codes with power decoding up to the Johnson radius.
\end{abstract}


\section{Introduction}

\noindent
An $\IntDegree$-interleaved Reed--Solomon (IRS) code is a direct sum of $\IntDegree$ Reed--Solomon (RS) codes of the same evaluation points.
The relevant metric is to consider \emph{burst errors}, where an error corrupts the same position in all $\IntDegree$ codewords.
In this metric, IRS codes can be decoded far beyond half the minimum distance of the constituent RS codes with high probability, and the decoding problem has achieved a lot of attention in the last decades, e.g.~\cite{krachkovsky1997decoding,bleichenbacher2003decoding,schmidt2009collaborative,wachterzeh2014decoding,coppersmith2003reconstructing,parvaresh2007algebraic}.
All of these decoders are \emph{partial}, in that they will fail for a few error patterns of any weight beyond half the minimum distance of the weakest constituent code.

Often \emph{homogeneous} IRS codes are considered for simplicity, i.e.~where the constituent RS codes all have the same rate $R$.
For $R > 1/3$ and $m>2$, the state of the art is still given by \cite{krachkovsky1997decoding} with a relative decoding radius of  $\frac{\IntDegree}{\IntDegree+1}(1-R)$, see \cref{fig:radii_comparison_m=2}.
This decoding radius was echoed in \cite{schmidt2009collaborative} with a better complexity by solving the classical \emph{key equations} for the constituent RS code simultaneously.

Power decoding \cite{schmidt2010syndrome} is a partial decoding algorithm for RS codes that can decode beyond half the minimum distance for $R \leq 1/3$, achieving roughly the same decoding radius as the Sudan list-decoder~\cite{sudan1997decoding}.
The idea is to generate several linearly independent key equations from one received word by a non-linear operation.
In~\cite{schmidt2007enhancing}, it was proposed to use power decoding in IRS codes to obtain several key equations from each constituent received word, thereby increasing the decoding radius for $R \leq 1/3$.
\cite{wachterzeh2014decoding} refined this approach by mixing the received words of the IRS code in the powering process.

Recently \cite{nielsen2015power}, an improved power decoding algorithm for RS codes was proposed which is able to decode up to the Johnson radius. The gain is obtained by introducing multiplicities, similar to the Guruswami--Sudan algorithm, resulting in more linearly independent key equations.

In this paper, we apply the improved power decoding of \cite{nielsen2015power} to IRS codes using the ideas of \cite{schmidt2007enhancing, wachterzeh2014decoding}.
We obtain a system of key equations and show how to efficiently solve them.
We argue that the algorithm will decode with high probability up to a relative distance $1 - R^{\frac m {m+1}}$, and
we support this using simulations for a wide variety of parameters.
For $\IntDegree>2$,\footnote{In \cite{parvaresh2007algebraic}, an interpolation-based algorithm is proposed for $\IntDegree=2$ that can achieve the same decoding radius as ours. It is claimed that it generalizes to $\IntDegree>2$. However, no root-finding algorithm is given for the general case. See \cref{ssec:comparison_other_IRS_decoders} for more details.} this is an asymptotical improvement for all rates $R$ over the previous best \cite{wachterzeh2014decoding,coppersmith2003reconstructing} (cf.~\cref{fig:radii_comparison_m=2} for the case $\IntDegree=3$).

The speed of our decoder depends on the sought performance: decoding up to relative distance $1 - R^{\frac m {m+1}} - \varepsilon$, the complexity is $\Osoft(n (1/\varepsilon)^{m \omega + 1})$, where $\Osoft$ means omission of logarithmic factors and $\omega$ is the exponent for matrix multiplication.
This matches or improves upon all previous algorithms whenever the decoding radius is comparable.

The decoding algorithm has been implemented in SageMath v7.5 \cite{stein_sagemath_????}, and the source code can be downloaded from \url{http://jsrn.dk/code-for-articles}.

\begin{figure}
\begin{center}
 \input{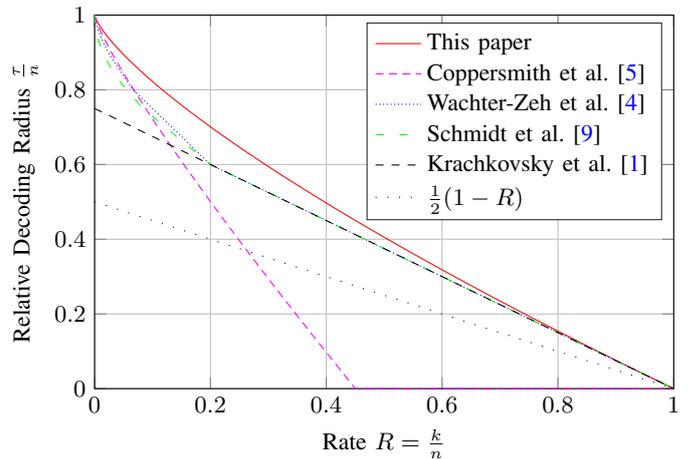}
\end{center}
\vspace{-0.3cm}
\caption{Comparison of Relative Decoding Radii for $\IntDegree=3$.}
\label{fig:radii_comparison_m=2}
\end{figure}

\section{Preliminaries}
\label{sec:preliminaries}

\noindent
Let $q$ be a prime power and $\Fq$ a finite field of size $q$.
\begin{definition}
Let $\alpha_1,\dots,\alpha_n \in \Fq$, $k<n$. The corresponding \emph{Reed--Solomon (RS) code} is the set
\begin{align*}
\RS(n,k) = \left\{ \left[f(\alpha_1),\dots,f(\alpha_n)\right] : f \in \Fq[x], \deg f <k \right\},
\end{align*}
where we call $f$ a \emph{message polynomial} and the $\alpha_i$'s \emph{evaluation points}.
An \emph{$\IntDegree$-Interleaved Reed--Solomon (IRS) code} is the direct sum of $\IntDegree$ RS codes $\RS(n,k_i)$ having the same evaluation points, i.e.,
\begin{align*}\def\arraystretch{0.9}
\IRS(n,k_1,\dots,k_\IntDegree)
&= \left\{ \begin{bmatrix} \c_1 \\ \vdots \\ \c_\IntDegree\end{bmatrix} : \c_i \in \RS(n,k_i) \right\} \in \Fq^{\IntDegree \times n} .
\end{align*}
A \emph{homogeneous} IRS code is an IRS code whose $\IntDegree$ constituent codes have the same dimension $k_i=k$.
\end{definition}

For simplifying of notation and analysis, we only consider homogeneous IRS codes in this paper.

As a metric for errors, we consider the \emph{burst error metric}: if $\vec r =\vec c +\vec e \in \Fq^{\IntDegree \times n}$ is received with $\vec c \in \IRS(n,k; m)$, then the \emph{error positions} are the non-zero columns of $\vec e$, i.e.~$\Eset = \bigcup_{i=1}^{\IntDegree} \{j : e_{i,j} \neq 0 \}$, and the \emph{number of errors} is $|\Eset|$.
Burst errors occur naturally, for instance, if we transmit with a symmetric $q^\IntDegree$-ary channel using the isomorphism $\Fq^{\IntDegree \times n} \simeq (\Fq^\IntDegree)^n$.

For a vector $\vec i = (i_1,\ldots, i_m) \in \ZZ_{\geq 0}^m$ we define its \emph{size} as $|\vec i| := \sum_t i_t$.
We denote by $\smaller$ the product partial order on $\ZZ_{\geq 0}^m$, i.e.~
$\i \smaller \j$ whenever $i_t \leq j_t$ for all $t$.
The number of vectors $\vec i \in \ZZ_{\geq 0}^m$ with $\len{\i} = \mu$ is given by $\tbinom{m+\mu-1}{\mu}$.
We will use the following well-known relations.
\begin{lemma}\label{lem:binomial_sums}
Let $m,t \in \ZZ_{>0}$. Then,
\begin{align*}
\sum_{\mu=0}^{t} \tbinom{m+\mu-1}{\mu} = \tbinom{m+t}{m}, \text{ and }
\sum_{\mu=0}^{t-1} \mu \tbinom{m+\mu-1}{\mu} = r \tbinom{m+t-1}{m+1}.
\end{align*}
\end{lemma}
We also introduce the following notational short-hands:
\begin{definition}
For $\a \in \Fq[x]^m$, and $\i,\j \in \ZZ_{\geq 0}^m$, we define
\begin{align*}
\a^\i := \prod_{t=1}^{m} a_t^{i_t}, \quad
\binom{\j}{\i} := \prod_{t=1}^{m} \binom{j_t}{i_t}.
\end{align*}
\end{definition}
We then have a vectorised binomial theorem:
\begin{lemma}\label{lem:generalized_binomial}
Let $r \in \ZZ_{>0}$, $\a,\b \in \Fq[x]^m$, and $\j \in \ZZ_{\geq 0}^m$. Then,
\begin{align*}
(\a+\b)^\j = \sum\limits_{\i \smaller \j} \binom{\j}{\i} \a^\i \b^{\j-\i}.
\end{align*}
\end{lemma}

\section{Key Equations}
\label{sec:new_decoder}

\noindent
In the following, we develop ``powered key equations'' for IRS codes similar to \cite{nielsen2015power}.
As in \cite{nielsen2015power}, we use the formulation of Power decoding introduced in \cite{nielsen_power_2014}.

Consider a decoding instance, i.e.~that $\vec r = \vec c + \vec e \in \Fq^{\IntDegree \times n}$ has been received with $\vec c \in \IRS(n, k; m)$ and $\Eset$ the corresponding error positions.
Let $\f = (f_1,\ldots,f_m) \in \Fq[x]_{< k}^m$ be the message polynomials corresponding to $\vec c$.

\begin{definition}\label{def:error_locator}
The \emph{error locator polynomial} is defined by
\begin{align*}
  \textstyle
\Lambda := \prod_{i \in \Eset} (x-\alpha_i).
\end{align*}
\end{definition}

\begin{definition}\label{def:Ri_G}
For $t=1,\dots,\IntDegree$, let $R_t \in \Fq[x]$ be the unique interpolation polynomial of degree $\deg R_t <n$ such that
\begin{align*}
R_t(\alpha_i) = r_{t,i} \quad \forall i=1,\dots,n.
\end{align*}
We write $\R = (R_1,\dots,R_\IntDegree)$.
Let $G \in \Fq[x]$ be given by
\begin{align*}
  \textstyle
G := \prod_{i=1}^{n} (x-\alpha_i).
\end{align*}
\end{definition}
Note that $\Lambda$ is not known at the receiver, but $\R$ and $G$ are. We can relate the polynomials as follows.

\begin{lemma}\label{lem:Omega}
Let $G$ and $\R$ be as in \cref{def:Ri_G}.
Then $G$ divides each element of $\Lambda(\f - \R)$.
\end{lemma}
\begin{proof}
For all $i=1,\dots,\IntDegree$ and $j=1,\dots,n$, we obtain
\begin{align*}
\big(\Lambda (f_i-R_i)\big)(\alpha_j) = \Lambda(\alpha_j) \cdot (-e_{i,j}) = \begin{cases}
0 \cdot (-e_{i,j}), &j \in \Eset \\
\Lambda(\alpha_j) \cdot 0, &\text{else}.
\end{cases}
\end{align*}
Thus, $(x-\alpha_j) \mid \Lambda (f_i-R_i)$ and the claim follows.
\end{proof}
In light of the above lemma, let
\[
  \OmegaVec := \Lambda (\f-\R)/G \in \Fq[x]^m .
\]
Following nomenclature from RS decoding, the $\OmegaVec$ is known as \emph{error evaluators}.
Note that the entries of $\OmegaVec$ each have degree less than $|\Eset|$.

We can now state the set of key equations that we will use for decoding.
It consists of $\binom{\IntDegree+ \PowParam}{\IntDegree}-1$ many relations between the unknown polynomials $\Lambda$, $\f$, and $\OmegaVec$.

\begin{theorem}[Key Equations]\label{thm:key_equation}
  Choose any $s, \ell \in \ZZ_{> 0}$ with $s \leq \ell$.
Let $\Lambda$, $G$, $\R$, $\f$, and $\OmegaVec$ be as above. Then, for all $\j \in \ZZ_{\geq 0}^\IntDegree$ with $1 \leq \len{\j} \leq \PowParam$, we have
\begin{align*}
\Lambda^\MultParam \f^\j
&= \sum\limits_{\i \smaller \j} \left[ \Lambda^{s-\len{\i}} \OmegaVec^\i \right] \left[ \binom{\j}{\i} \R^{\j-\i} G^{\len{\i}}\right], && \len{\j}<\MultParam \\
\Lambda^\MultParam \f^\j &\equiv \sum\limits_{\substack{\i \smaller \j \\ \len{\i} < \MultParam}}
\left[ \Lambda^{s-\len{\i}} \OmegaVec^\i\right]
\left[ \binom{\j}{\i} \R^{\j-\i} G^{\len{\i}}\right]
  \mod G^\MultParam, &&\len{\j} \geq \MultParam.
                        \label{eq:key_eq_2}
\end{align*}
\end{theorem}

\begin{proof}
Using \cref{lem:generalized_binomial}, we can write
\begin{align*}
\Lambda^s \f^\j = \Lambda^s \big((\f-\R) + \R\big)^\j = \Lambda^s \sum_{\i \smaller \j} \binom{\j}{\i} (\f-\R)^\i \R^{\j-\i}.
\end{align*}
For $\len{\i}<s$, we rewrite
\begin{align*}
\Lambda^s (\f-\R)^\i = \Lambda^{s-\len{\i}} [\Lambda(\f-\R)]^\i = \Lambda^{s-\len{\i}} \OmegaVec^\i G^{\len{\i}}.
\end{align*}
If $\len{\i} \geq s$, we decompose $\i= \i_s+\i'$, where $\i_s,\i' \in \ZZ_{\geq 0}^\IntDegree$ with $\len{\i_s} = s$. Using this notation, we obtain
\begin{align*}
\Lambda^s (\f-\R)^\i = [\Lambda(\f-\R)]^{\i_s} (\f-\R)^{\i'} = G^s \OmegaVec^{\i_s} (\f-\R)^{\i'},
\end{align*}
so all summands $\i$ with $\len{\i} \geq s$ are zero modulo $G^s$.
Hence, the equation and congruence above hold.
\end{proof}
We use the following abbreviations in the remaining sections.
\begin{align*}
\Psi_\j := \Lambda^\MultParam \f^\j, \;
\Lambda_\i := \Lambda^{s-\len{\i}} \OmegaVec^\i, \text{ and }
A_{\i,\j} := \tbinom{\j}{\i} \R^{\j-\i} G^{\len{\i}}.
\end{align*}

\section{Solving the Key Equations}
\label{sec:solving}

\noindent
The key equations of \cref{thm:key_equation} are non-linear relations between the unknowns $\Lambda, \f$ and $\OmegaVec$ and therefore a priori difficult to solve.
The approach, as for classical key equation decoding, is to relax the non-linear relations to---seemingly much weaker---linear relations, and then hope that a minimal solution to these is the sought.
The linear problem that we will relax into is a very general variant of \Pade approximations:

\begin{problem}
  \label{prob:SimHermitePade}[Simultaneous Hermite \Pade approximation]
  Given $S_{i,j}, G_j \in \Fq[x]$ for all $i \in I, j \in J$ for index sets $I,J$ as well as degree bounds $N_i, T_j \in \ZZ_{\geq 0}$, compute, if it exists, $\lambda_i, \psi_j \in \Fq[x]$, not all zero and such that
  \[
    \sum_{i \in I} \lambda_i S_{i,j} \equiv \psi_j \mod G_j, \quad\forall j \in J \ ,
  \]
 as well as $\deg \lambda_i < N_i$ and $\deg \psi_j < T_j$ for all $i \in I, j \in J$.
\end{problem}

Such $\lambda_i, \psi_j$ is a \emph{solution} to the \Pade approximation.
If furthermore $\lambda_{\vec 0} \neq 0$ and has minimal degree for a specific $\vec 0 \in I$, we say that it is a \emph{minimal solution}.

\begin{proposition}
  \label{prop:decoding_shpade}
  Let $\tau \geq \deg(\Lambda)$ and $\Lambda_\i, \Psi_\j$ and $A_{\i,\j}$ be defined as in \cref{sec:new_decoder}.
  Then $\Lambda_\i, \Psi_\j$ is a solution to the simultaneous Hermite \Pade approximation given by:
  \begin{align*}
    I                 & = \{ \i \in \ZZ_{\geq 0}^m, |\i| < s \} & J & = \{ \j \in \ZZ_{\geq 0}^m,  |\j| \leq \ell \} \\
    S_{\i, \j}          & = A_{\i,\j} & G_\j              & = \left\{
                                                          \begin{array}{ll}
                                                            x^{\tau s + |\j|(n - 1) + 1} & |\j| < s \\
                                                            G^s           & |\j| \geq s
                                                          \end{array}\right. \\
    N_\i              & = \tau s - |\i|+1             & T_\j &= \tau s + |\j|(k-1) + 1 \ ,
  \end{align*}
  and leader index $\vec 0$ (all-zero vector).
\end{proposition}
\begin{proof}
  That $\Lambda_\i, \Psi_\j$ satisfy the degree constraints follows from $\tau \geq \deg(\Lambda)$ and $\deg \f^\j \leq |\j|(k-1)$.
  For $|\j| \geq s$ the \Pade approximation congruence is that of \cref{thm:key_equation}, and for $|\j| < s$, the congruence is satisfied since it is an equality in \cref{thm:key_equation}.
\end{proof}

\begin{remark}
  The modulus $x^{\tau s + |\j|(n - 1) + 1}$ in \cref{prop:decoding_shpade} arise as $1 + \deg\big(\sum_{|\i| < s} \lambda_\i A_{\i,\j}\big)$ for $|\j| < s$.
\end{remark}
In other words, solving the simultaneous Hermite \Pade approximation given by \cref{prop:decoding_shpade} will, up to a scalar multiple, recover $\Lambda, \f$ and $\OmegaVec$ if $\Psi_\j$ and $\Lambda_\i$ is the \emph{minimal} solution.
The decoding algorithm is built around exactly this expectation.
Formally, the Power-IRS decoding algorithm is as follows:
\begin{enumerate}
  \item Compute a \emph{minimal} solution $(\lambda_\i)_\i, (\psi_\j)_\j$ to the simultaneous Hermite \Pade approximation given by \cref{prop:decoding_shpade}, minimising $\deg \lambda_{\vec 0}$.
  \item Let $f_t := \psi_{\ve u_t}/\lambda_{\ve 0}$ for each $t=1,\ldots,m$, where the $\ve u_t$ is the vector with 0 everywhere but a 1 at index $t$.
  If these are not all polynomials of degree less than $k$, declare $\mathsf{fail}$.
  \item If $\wt(\vec r -  \vec c) \leq \tau$, return $\vec c$, where $\vec c$ is the codeword of the message polynomials $(f_1,\ldots,f_m)$ and $\wt(\vec x)$ is the number of non-zero columns of $\vec x$. Otherwise declare $\mathsf{fail}$.
\end{enumerate}
We discuss in later sections when we can expect this algorithm to work.
For now we turn to the complexity of the decoder.
Solving variants of \Pade approximation problems has a long and lustrous history in both coding theory and computer algebra.
In particular, \cite{rosenkilde2017SimHermPade} gives the currently fastest algorithm for our case.
The very general algorithm of \cite{jeannerod_fast_2016} can also be used but is slightly slower.
See also these papers for a discussion of the history of computing \Pade approximations.
We obtain:
\begin{proposition}\label{prop:complexity}
  Given a homogeneous IRS code $\IRS$ and $s, \ell \in \ZZ_{> 0}$, the Power-IRS decoding algorithm can be performed in $\Osoft\!\left(n s \CountLeq \ell^{\omega-1} \CountLt{s} \right)$ operations in $\Fq$.
\end{proposition}

\begin{proof}
  The $\CountLeq \ell \CountLt{s}$ values $A_{\i,\j}$ can easily be computed in the target cost once we know $\R^\i \modop G^s$ for all $\i$ with $|\i| \leq \ell$.
  These are computed using dynamic programming in time $\Osoft(sn)$.
  Solving the \Pade approximation is the only other expensive step and has the target cost by using \cite{rosenkilde2017SimHermPade}.
\end{proof}

\section{Decoding Performance}

\noindent
One cannot hope to improve upon single-RS decoding in the worst case scenario: the sent codeword and error could be the same in each constituent code, so the interleaving gives no new information.
Similarly, list decoding IRS codes cannot go beyond list decoding RS codes.

As previous IRS decoders, Power-IRS sidesteps this issue by being a \emph{partial decoder} whose capability is characterised by two concepts: 1) the number of errors we expect to decode, which we will call the ``\emph{decoding radius}''; and 2) the \emph{probability of success} for a given number of errors, assuming uniformly random error patterns of a given weight.

Solving the simultaneous Hermite \Pade approximation of \cref{prop:decoding_shpade} will lead to finding $\vec \Lambda, \vec f$ and $\OmegaVec$ if there are not other solutions to the approximation.
In \cref{sec:decoding_radius} we define the decoding radius as just below the number of errors when ``generic'' solutions are sure to appear, effectively ruling out that we find the special solution.
In \cref{sec:Failure_Probability} we prove that success or failure is dictated only by the error, and not the sent codeword.
We then turn to the success probability: we have no formal bounds, but we discuss what we believe are the main contributing factors  and present a conjecture.
The conjecture is backed by simulations in \cref{sec:simulations}.

\subsection{Decoding Radius}
\label{sec:decoding_radius}

\begin{lemma}
  \label{lem:shpade_has_solution}
  Consider a simultaneous Hermite \Pade approximation (\cref{prob:SimHermitePade}).
  The approximation is guaranteed to have at least two $\Fq$-linearly independent non-zero solutions whenever
  \[
    \textstyle
    \sum_{i \in I} N_i > 1 + \sum_{j \in J}(\deg G_j - T_j)
  \]
\end{lemma}
\begin{proof}
  The degree restrictions on the remainders $\psi_j$ can be considered as $\Fq$-linear restrictions in the coefficients of the $\lambda_i$: that $\deg \psi_j < T_j$ implies that the coefficients to $x^{T_j},\ldots,x^{\deg G_j - 1}$ are zero in the corresponding $\Fq[x]$-linear expression in the $\lambda_i$.
  This gives $\sum_{i \in I} N_i$ indeterminates and at least $\sum_{j \in J}(\deg G_j - T_j)$ restrictions over $\Fq$.
\end{proof}

\begin{theorem}\label{thm:decoding_radius}
  The simultaneous Hermite \Pade approximation of \cref{prop:decoding_shpade} is guaranteed to have a solution which is $\Fq$-linearly independent from $\Lambda_\i, \Psi_\j$ whenever $\tau > \tauRP$, where
  {\small
  \begin{align*}
    &\tauRP :=\\
    &n\left[ 1 - \tfrac{s \CountLt s - \WeightedLt s}{s\CountLeq \ell} \right]
          - \tfrac m {m+1} \tfrac \ell s (k-1) - \tfrac 1 s \left[1 - \tfrac 1 {\CountLeq \ell}\right]
  \end{align*}
  }
\end{theorem}

\begin{proof}
See \cref{app:technical_proofs}.
\end{proof}

Following the discussion above, we call $\tauRP$ the \emph{decoding radius} of the Power-IRS decoder.
In the following asymptotic analyses, $O_{R,m}$ means that $R$ and $m$ are considered constants (i.e.~hidden).

\begin{theorem}\label{thm:asymptotic_relative_decoding_radius}
Let $(\PowParam_i,\MultParam_i) = (i, \round{\gamma i}+1)$ for $i \in \ZZ_{>0}$, where $\gamma = \sqrt[\IntDegree+1]{\frac{k-1}{n}}$. Then,
\begin{align*}
\tau_\mathrm{new}(\PowParam_i,\MultParam_i) = n \Big( 1-\left(\tfrac{k-1}{n}\right)^{\frac{\IntDegree}{\IntDegree+1}} - O_{R,m}\!\left( \tfrac{1}{i} \right) \Big).
\end{align*}
\end{theorem}

\begin{proof}
See \cref{app:technical_proofs}.
\end{proof}

\begin{corollary}
  For a fixed code and any constant $\varepsilon > 0$, we can choose $s, \ell \in O_{R,m}(1/\varepsilon)$ such that $\tauRP \geq n(1 - R^{\frac m {m+1}} - \varepsilon)$ where $R = \tfrac k n$.

  In this case, the algorithm has complexity $\Osoft\big(n (1/\varepsilon)^{m\omega + 1}\big)$.
\end{corollary}

\subsection{Expecting Successful Decoding}\label{sec:Failure_Probability}

Now we turn to the probability that Power-IRS will succeed.
We start with the observation that decoding success is independent of the sent codeword.

\begin{theorem}
The success of decoding $\r=\c+\e$ depends only on the error $\e$.
\end{theorem}

\begin{proof}
The proof closely resembles the one of \cite[Proposition~6]{nielsen2015power}.
It suffices to show that if decoding $\r$ fails, then decoding $\r+\chat$ for any $\chat \in \IRS$ also fails. Let $\chat$ correspond to the message polynomial $\hat{\f}$. If decoding $\r$ fails, there is a solution $\lambda_\i, \psi_\j$ to the \Pade approximation of \cref{prop:decoding_shpade} for received word $\r$ with $\lambda_\0 \neq \Lambda^s$ and $\deg \lambda_\0 \leq \deg \Lambda^s$. We claim that $\psihat_\j := \sum_{\i \smaller \j} \tbinom{\j}{\i} \hat{f}^\i \psi_{\j-\i}$ and $\lambdahat_\i := \lambda_\i$ is a solution to the approximation problem for received word $\r+\chat$, so that decoding $\r + \chat$ also fails.
Let $\Rhat := \R+\fhat$.
Then the required congruences are satisfied since
\begin{align*}
\sum\limits_{\substack{\i \smaller \j \\ \len{\i}<\MultParam}} \lambdahat_\i \tbinom{\j}{\i} \Rhat^{\j-\i} G^{\len{\i}}
&= \sum\limits_{\substack{\i \smaller \j \\ \len{\i}<\MultParam}} \sum\limits_{\myindex \smaller \j-\i} \lambda_\i \underset{= \, \binom{\j}{\myindex} \binom{\j-\myindex}{\i}}{\underbrace{\tbinom{\j}{\i} \tbinom{\j-\i}{\myindex}}} \R^{\j-\i-\myindex} \fhat^\myindex G^{\len{\i}}  \\
&= \sum\limits_{\myindex \smaller \j} \tbinom{\j}{\myindex} \fhat^\myindex \sum\limits_{\substack{\i \smaller \j-\myindex \\ \len{\i}<\MultParam}} \lambda_\i  \tbinom{\j-\myindex}{\i} \R^{\j-\myindex-\i} G^{\len{\i}} \\
& \equiv \sum\limits_{\myindex \smaller \j} \tbinom{\j}{\myindex} \fhat^\myindex \psi_{\j-\myindex} = \psihat_\j \mod G_\j
\end{align*}
Also, $\deg \psihat_\j \leq \min_\myindex \{ \deg \psi_{\j-\myindex}+\len{\myindex} (k-1) \} \leq \tau \MultParam + \len{\j} (k-1) = T_\j$.
\end{proof}

Turning to the probability of success.
There are two mechanisms we should expect cause Power-IRS to fail on reception of $\vec r = \vec e + \vec c$: 1) the simultaneous Hermite \Pade approximation has a ``spurious'', non-coding theoretic related solution; and 2) there is a codeword $\vec c' \neq \vec c$ with $\wt(\vec r - \vec c') \leq \wt(\vec r - \vec c)$.
The former probability can be characterised for ``random'' \Pade approximations (i.e. random $S_{i,j}$ polynomials), and the latter can easily be estimated using classical coding theory.
Both are exponentially decaying in the value $(\tauRP - \varepsilon)$, where $\varepsilon$ is the number of errors.
We expected that these were essentially the only contributing factors to decoding failure; our simulations indicate that the observed failure rate are often, but not always, very close to the union bound of the above two probabilities.

This important question needs more investigation.
For now, based on the observed failure rates, and the formal bounds on success probability for Power decoding RS codes, we conjecture that the probability that Power-IRS succeeds on an error $\vec e$ with $\varepsilon$ non-zero columns can be lower bounded by $1 - q^{-b(\tauRP - \varepsilon)}$ for some constant $b > 1$ that depends on the code and decoder parameters $n$, $k$, $\IntDegree$, $\PowParam$, and $\MultParam$.

\section{Simulations}
\label{sec:simulations}

\noindent
We implemented the decoding algorithm in SageMath v7.5 \cite{stein_sagemath_????} and observed the decoding behaviour on random errors for a range of parameters.
The source code can be downloaded from \url{http://jsrn.dk/code-for-articles}.

Some of the results are outlined in \cref{tab:simulation_results_short}.
For instance, the $\IRS(257,86;2)$ code over $\F_{257}$ can be decoded by previous algorithms up to $114$ errors. Our algorithm is able to decode up to $124$ errors by choosing $(\ell,s) = (4,3)$ and we observed failure in $\approx 1.1 \cdot 10^{-5}$ fraction of the trials.

Over $\F_{17}$, we can decode up to $13$ errors with an observed failure rate of $\approx 1.9 \cdot 10^{-3}$ using our decoder for the code $\IRS(17,3;5)$ with parameters $(5,3)$ (instead of $11$ using \cite{wachterzeh2014decoding}).
Note that this is very close to $n-k = 14$.

We also compared our decoder to the example in \cite[Table~1]{wachterzeh2014decoding}: An $\IRS(16,2;3)$ code over $\F_{17}$. The decoder in \cite{wachterzeh2014decoding} can decode up to $12$ errors with observed failure rate $6.2 \cdot 10^{-3}$. By choosing the parameters of our decoder $(3,2)$, we observe a lower failure rate $9.1 \cdot 10^{-5}$ for the same number of errors. By further increasing the parameters, we decode one more error with $\approx 90\%$ probability.

In general, the observed failure rate $\hat{\mathrm{P}}_\mathrm{fail}$ rapidly decreases with the number of errors, backing up the conjecture of the previous section.
Also, whenever the number of errors was above the radius, decoding always failed.

\begin{table}[h]
\caption{Simulation results. Code parameters $q,n,k,m$ , Decoder parameters $(\ell,s)$, Number of samples $N$. Observed Failure Rate $\hat{\mathrm{P}}_\mathrm{fail}$. For comparison: Decoding radii $\tauWZB$ of \cite{wachterzeh2014decoding} and $\tauKL$ of \cite{krachkovsky1997decoding} for the chosen parameters.} 
\label{tab:simulation_results_short}
\centering
\renewcommand{\arraystretch}{1.3}
\setlength{\tabcolsep}{0.5em}
\scriptsize
\begin{tabular}{c||c|c|c||c|c|c||c|c||c}
$q$   & $n$   & $k$  & $m$ & $(\ell,s)$ & $\tau_\mathrm{new}$ & $\hat{\mathrm{P}}_\mathrm{fail}$  & $\tauWZB$ & $\tauKL$ & $N$    \\
\hline
$257$ & $257$ & $86$ & $2$ & $(3,2)$    & $120$ & $0$          & $114$ & $114$ & $10^6$ \\
      &       &      &     & $(4,3)$    & $124$ & $1.1 \cdot 10^{-5}$ & $114$ & $114$ & $10^6$ \\
\hline
$43$  & $43$  & $18$ & $2$ & $(4,3)$    & $18$  & $5.4 \cdot 10^{-4}$ &  $16$ &  $16$ & $10^6$ \\
\hline
$17$  & $17$  & $3$  & $2$ & $(3,2)$    & $11$  & $1.9 \cdot 10^{-3}$ &  $10$ &   $9$ & $10^6$ \\
      &       &      & $4$ & $(4,3)$    & $13$  & $2.8 \cdot 10^{-2}$ &  $10$ &   $9$ & $10^4$ \\
      &       &      & $5$ & $(5,3)$    & $13$  & $1.9 \cdot 10^{-3}$ &  $10$ &   $9$ & $10^4$ \\
\hline
$17$  & $16$  & $2$  & $3$ & $(3,2)$    & $12$  & $9.1 \cdot 10^{-5}$ &  $12$ &  $10$ & $10^6$ \\
    &       &      &     & $(6,3)$    & $13$  & $1.0 \cdot 10^{-1}$ &  $12$ &  $10$ & $10^4$ \\
\hline
$16$  & $16$  & $3$  & $3$ & $(2,1)$    & $10$  & $0$         &  $10$ &   $9$ & $10^6$ \\
      &       &      &     & $(3,2)$    & $11$  & $2.1 \cdot 10^{-5}$ &  $10$ &   $9$ & $10^6$ \\
\end{tabular}
\end{table}

\section{Comparison to Related Work}

\noindent
We now compare the asymptotic decoding performance of Power-IRS to other IRS decoders and Folded RS codes.

\subsection{Other Decoding Algorithms for Interleaved RS Codes}
\label{ssec:comparison_other_IRS_decoders}

Using the strict generalized Bernoulli inequality, we obtain
\begin{align*}
\tfrac{\IntDegree}{\IntDegree+1} (1-R) < 1-R^\frac{\IntDegree}{\IntDegree+1}
\end{align*}
for all rates $R \in (0,1)$ and interleaving degrees $\IntDegree \geq 1$. Hence, our algorithm improves upon the algorithms in \cite{krachkovsky1997decoding,schmidt2009collaborative} at all rates.
For $R>\tfrac{1}{3}$ and $\IntDegree>2$, this is the first improvement since 1997 \cite{krachkovsky1997decoding}.
For $R < \tfrac{1}{3}$, the previous best are the CS~\cite{coppersmith2003reconstructing} and WZB~\cite{wachterzeh2014decoding} algorithms (the WZB algorithm strictly improves upon the KL \cite{krachkovsky1997decoding} and SSB \cite{schmidt2007enhancing} algorithm).
The relative decoding radius of the CS algorithm is $\tfrac{\tauCS}{n} \to 1-R^\frac{\IntDegree}{\IntDegree+1}-R$ $(n \to \infty)$ \cite[Theorem~1]{coppersmith2003reconstructing},
which is strictly smaller than our $1-R^\frac{\IntDegree}{\IntDegree+1}$.
The WZB algorithm is a special case of Power-IRS where $s = 1$.
As demonstrated in \cref{fig:radii_comparison_m=2} on the first page for $m=3$, the additional degree of freedom of $\MultParam$ greatly increases the decoding capability.
Note also that for $\IntDegree=1$, Power-IRS reduces to the algorithm in \cite{nielsen2015power} and decodes to the Johnson radius as in \cite{guruswami1998improved}.

In his PhD thesis \cite[Chapter~2]{parvaresh2007algebraic}, Parvaresh described a decoding algorithm for interleaved RS codes based on multi-variate interpolation, which can be seen as a generalization of the Guruswami--Sudan algorithm.
The existence of a suitable $(m+1)$-variate interpolation polynomial is guaranteed if the relative number of errors is smaller than $1-R^{\frac{\IntDegree}{\IntDegree+1}}$.
However, the root-finding step is only described for the case $\IntDegree=2$ (cf.~\cite[Section~2.5]{parvaresh2007algebraic}).
To the best of our knowledge, no subsequent work on this topic was published by the author, so the case $\IntDegree>2$ remains an open problem.

Cohn and Heninger \cite{cohn2013approximate} proposed an algorithm for noisy multi-polynomial reconstruction, which is closely related to decoding IRS codes.
More precisely, \cite[Theorem~3]{cohn2013approximate} implies that decoding is successful if the fraction of errors is less than $1-R^{\frac{\IntDegree}{\IntDegree+1}}$, given that certain polynomials, which depend on the algorithm's input, are algebraically independent.
The paper does not state under which conditions this so-called algebraic independence hypothesis holds.
In future work, the performance of this decoder must be evaluated more thoroughly.
In any case, our algorithm seems to result in a much smaller complexity due to the non-existence of the heavy root-finding step, which requires resultant or Gr\"obner-basis computations in \cite{cohn2013approximate}.

\subsection{Folded RS Codes}

\emph{$m$-folded Reed--Solomon codes} are very similar to IRS codes and are interesting since they allow list-decoding beyond the Johnson radius \cite[Theorem~4.4]{guruswami2008explicit}: for any $\varepsilon>0$, there is a family of $m$-folded RS codes of rate $R$ for which a fraction of $1-R-\varepsilon$ errors can be list-decoded, where $m \in O(1/\varepsilon^2)$.

In comparison, our results indicate that for any $\varepsilon>0$, any $m$-interleaved RS code of rate $R$ with $m = \tfrac{\log(1/(R+\varepsilon))}{\log(1+\varepsilon/R)} \leq \tfrac{\log(1/R)}{\varepsilon/R} \in O(1/\varepsilon)$ can correct a fraction of
\begin{align*}
1-R^\frac{\IntDegree}{\IntDegree+1} = 1-R^{\frac{\log(R+\varepsilon)}{\log(R)}} = 1-R-\varepsilon
\end{align*}
errors.
Hence the decoding radius of Power-IRS convergences much faster to capacity $1-R$ with respect to the interleaving degree (equivalently the field size $q^m$).
It should be kept in mind that the algorithm in \cite{guruswami2008explicit} is a list decoder, so it guarantees to return a list of codewords containing the sent codeword, which is a much stronger guarantee than Power-IRS.

Another related capacity-achieving construction is univariate multiplicity codes \cite{kopparty_list-decoding_2012} (also called Derivative codes) which has the same convergence rate $m \approx O(1/\varepsilon^2)$ as Folded RS codes.

\section{Conclusion}

\noindent
We have introduced a new partial decoding algorithm for $m$-Interleaved Reed--Solomon codes that can decode, with seemingly high probability, up to $n(1 - R^{\frac m {m+1}} - \varepsilon)$ errors with complexity quasi-linear in $n$ and polynomial in $1/\varepsilon$.
Our decoding radius is motivated by essentially assuming that the underlying linear system is random, but simulations on a variety of parameters strongly back up this value.
Formally bounding the success probability remains an open problem.

\section*{Acknowledgement}

We would like to thank Hannes Bartz for making us aware of \cite{parvaresh2007algebraic} and Vincent Neiger for pointing us at \cite{cohn2013approximate}.

\bibliographystyle{IEEEtran}
\bibliography{main}

\appendix

\subsection{Technical Proofs}
\label{app:technical_proofs}

\begin{proof}[Proof of \cref{thm:decoding_radius}]
By \cref{lem:shpade_has_solution}, we are ensured a solution different from $\Lambda_\i, \Psi_\j$ if $\sum_{|\i| < s} N_\i > \sum_{\j \leq \ell}(\deg G_\j - T_\j) + 1$, where $N_\i, T_\j, G_\j$ are as in \cref{prop:decoding_shpade}.
    We get {\small
    \begin{align*}
      &\sum_{|\i| < s} N_\i = \sum_{|\i| < s}(\tau s - |\i| + 1) = (\tau s + 1) \CountLtt{s} - \WeightedLtt{s}, \\
      &\sum_{|\j| \leq \ell}(\deg G_\j - T_\j) = \\
      &\quad \sum_{|\j| < s}\big(|\j|(n-k)\big) + \sum_{s \leq |\j| \leq \ell} \left(s(n - \tau) - |\j|(k-1) - 1\right) = \\
      &\quad (n-k)\WeightedLtt{s} + \left(s(n-\tau)-1\right)\left(\CountLeqt{\ell} - \CountLtt{s}\right) \\
          &\qquad\qquad - (k-1)\left(\WeightedLeqt{\ell} - \WeightedLtt{s}\right).
    \end{align*}
    The inequality becomes:
    \begin{align*}
      s \tau \CountLeqt{\ell} &> n\left[\WeightedLtt{s} + s\left(\CountLeqt{\ell} - \CountLtt{s}\right)\right] \\
                               & \quad - (k-1)\WeightedLtt \ell - \CountLeqt \ell + 1.
    \end{align*}
	}
\end{proof}

\begin{lemma}\label{lem:binomial_ratio_convergence}
Let $\IntDegree \in \ZZ_{>0}$ and $\gamma \in (0,1)$ be fixed.
Then for $i \in \ZZ_{>0}$ going to infinity, we have
\begin{align*}
\tfrac{\binom{\IntDegree+\round{\gamma i}}{\IntDegree}}{\binom{\IntDegree+i}{\IntDegree}} = \gamma^\IntDegree + O_{\gamma,m}\left( \tfrac{1}{i} \right).
\end{align*}
\end{lemma}
\begin{proof}
Using Stirling's formula $(\ast)$ and
{\small
\begin{align}
&1 \leq \sqrt{\tfrac{(\IntDegree+\gamma i)i}{(\IntDegree+i)\gamma i}} = \sqrt{1+\tfrac{(1-\gamma)m}{\gamma i}} \leq 1+\tfrac{(1-\gamma)m}{\gamma i}, \label{eq:convergence_sqrt} \\
&\left| e^x-\left( 1+\tfrac{x}{i}\right)^i\right| = O\left( \tfrac{1}{i} \right) \quad \forall \, x \in \mathbb{R}, \label{eq:convergence_euler} \\
&\left(\tfrac{\IntDegree+\gamma i}{\IntDegree+i}\right)^\IntDegree = \gamma^\IntDegree+\sum\limits_{j=1}^{\IntDegree} \tbinom{\IntDegree}{j} \gamma^{m-j} \left(\tfrac{(1-\gamma)m}{m+i}\right)^{j} = \gamma^\IntDegree+O_{\gamma,m}\left( \tfrac{1}{i} \right), \label{eq:convergence_power}
\end{align}}
we obtain
\vspace{-0.5cm}
\begin{align*}
&\tfrac{\binom{\IntDegree+\round{\gamma i}}{\IntDegree}}{\binom{\IntDegree+i}{\IntDegree}}
= \tfrac{(\IntDegree+\round{\gamma i})! i!}{\round{\gamma i}! (\IntDegree+i)!} \overset{(\ast)}{=} \overset{\overset{\eqref{eq:convergence_sqrt}}{=} 1+O(\frac{1}{i})}{\overbrace{\sqrt{\tfrac{(\IntDegree+\gamma i)i}{(\IntDegree+i)\gamma i}}}} \cdot \tfrac{(\IntDegree+\gamma i)^{\IntDegree+\gamma i} i^i}{(\gamma i)^{\gamma i} (\IntDegree+i)^{\IntDegree+i}} \\
& \quad \quad \quad \quad \quad \quad \quad \quad \quad \quad  \quad\cdot \underset{= 1}{\underbrace{e^{(\IntDegree+\gamma i)+i-\gamma i - (\IntDegree+i)}}} + O\left( \tfrac{1}{i} \right) \\
&= \underset{\underset{\eqref{eq:convergence_power}}{=} \gamma^\IntDegree+O(\frac{1}{i})}{\underbrace{\left(\tfrac{\IntDegree+\gamma i}{\IntDegree+i}\right)^\IntDegree}} \cdot \underset{\underset{\eqref{eq:convergence_euler}}{=} e^\IntDegree+O(\frac{1}{i})}{\underbrace{\tfrac{(\IntDegree+\gamma i)^{\gamma i}}{(\gamma i)^{\gamma i}}}} \cdot \underset{\underset{\eqref{eq:convergence_euler}}{=} e^{-\IntDegree}+O(\frac{1}{i})}{\underbrace{\tfrac{(i)^{i}}{(\IntDegree+i)^i}}} + O\left( \tfrac{1}{i} \right) = \gamma^\IntDegree + O\left( \tfrac{1}{i} \right),
\end{align*}
which proves the claim.
\end{proof}

\begin{proof}[Proof of \cref{thm:asymptotic_relative_decoding_radius}]
We choose $(\PowParam_i,\MultParam_i)$ as in the theorem statement. Using \cref{lem:binomial_ratio_convergence}, we obtain
\begin{align*}
\tfrac{\tauRP}{n}
  &= 1-\Big[1+\underset{= \, 1-O\left(\frac{1}{i}\right)}{\underbrace{\big(1-\tfrac{1}{\MultParam_i}\big)}} \tfrac{\IntDegree}{\IntDegree+1}\Big] \underset{\frac{\binom{\IntDegree+\round{\gamma i}}{\IntDegree}}{\binom{\IntDegree+i}{\IntDegree}} = \, \gamma^\IntDegree+O\left(\frac{1}{i}\right)}{\underbrace{\tfrac{\binom{\IntDegree+\MultParam_i-1}{\IntDegree }}{\binom{\IntDegree+\PowParam_i}{\IntDegree }}}\phantom{XXXXX}} \\
&\quad - \tfrac{\IntDegree}{\IntDegree+1} \underset{\substack{= \, \gamma^{-1}\\ + O\left( \frac{1}{i} \right)}}{\underbrace{\tfrac{\PowParam_i}{\MultParam_i}}}\tfrac{k-1}{n} 
- \underset{= \, O\left(\tfrac{1}{i}\right)}{\underbrace{\tfrac{1}{\MultParam_i}}} \underset{= \, 1-O\left(\frac{1}{i}\right)}{\underbrace{\left[1 - \tfrac 1 {\CountLeq \ell}\right]}} \\
&= 1+\tfrac{m}{m+1} \underset{= \, 0}{\underbrace{\Big( \gamma^\IntDegree-\gamma^{-1} \tfrac{k-1}{n} \Big)}} - \gamma^\IntDegree - O\left( \tfrac{1}{i} \right) \\
&= 1 - \left(\tfrac{k-1}{n}\right)^{\frac{\IntDegree}{\IntDegree+1}} - O\left( \tfrac{1}{i} \right),
\end{align*}
which proves the claim.
\end{proof}

\end{document}

%% file: defs.tex
\usepackage{etex}
\usepackage{pgfplots}
\pgfplotsset{compat=newest}
\usepackage{tikz}
\usetikzlibrary{arrows,matrix,positioning}
\usepackage[utf8]{inputenc}
\usepackage[innermargin=0.545in,outermargin=0.545in,top=0.545in,bottom=1in]{geometry}
\usepackage[english]{babel}
\usepackage[T1]{fontenc}
\usepackage{epsfig}
\usepackage{amsmath, amssymb, amsbsy}
\usepackage{mathdots}
\usepackage{xspace}
\usepackage[noend]{algpseudocode}
\usepackage{algorithm}
\usepackage{algorithmicx}
\usepackage{color}
\usepackage{cite}
\usepackage{booktabs}
\usepackage{verbatim}
\usepackage[colorlinks=true,citecolor=blue,linkcolor=blue,urlcolor=blue]{hyperref}
\usepackage{url}
\usepackage{lipsum}
\usepackage{enumitem}
\usepackage{colortbl}

\newtheorem{theorem}{Theorem}
\newtheorem{lemma}[theorem]{Lemma}
\newtheorem{corollary}[theorem]{Corollary}
\newtheorem{proposition}[theorem]{Proposition}

\newtheorem{problem}[theorem]{Problem}

\newtheorem{remark}[theorem]{Remark}
\newtheorem{definition}{Definition}

\usepackage[final,tracking=true,kerning=true,spacing=true,factor=1100,stretch=10,shrink=20]{microtype}

\algrenewcommand\alglinenumber[1]{{\scriptsize#1}}
\algrenewcommand\algorithmicrequire{\textbf{Input:}}
\algrenewcommand\algorithmicensure{\textbf{Output:}}

\usepackage{varioref}        
\usepackage{fancyref}

\makeatletter
\def\mkfancyprefix#1#2{%
\expandafter\def\csname fancyref#1labelprefix\endcsname{#1}%
\begingroup\def\x{\endgroup\frefformat{plain}}%
    \expandafter\x\csname fancyref#1labelprefix\endcsname
    {\MakeLowercase{#2}\fancyrefdefaultspacing##1}%
\begingroup\def\x{\endgroup\Frefformat{plain}}%
    \expandafter\x\csname fancyref#1labelprefix\endcsname
    {#2\fancyrefdefaultspacing##1}%
\begingroup\def\x{\endgroup\frefformat{vario}}%
    \expandafter\x\csname fancyref#1labelprefix\endcsname
    {\MakeLowercase{#2}\fancyrefdefaultspacing##1##3}%
\begingroup\def\x{\endgroup\Frefformat{vario}}%
    \expandafter\x\csname fancyref#1labelprefix\endcsname
    {#2\fancyrefdefaultspacing##1##3}%
}
\makeatother
\fancyrefchangeprefix{\fancyrefeqlabelprefix}{eqn}
\mkfancyprefix{ssec}{Section}
\mkfancyprefix{tbl}{Table}
\mkfancyprefix{thm}{Theorem}
\mkfancyprefix{lem}{Lemma}
\mkfancyprefix{cor}{Corollary}
\mkfancyprefix{prop}{Proposition}
\mkfancyprefix{prob}{Problem}
\mkfancyprefix{conj}{Conjecture}
\mkfancyprefix{alg}{Algorithm}
\mkfancyprefix{inv}{Invariant}
\mkfancyprefix{ex}{Example}
\mkfancyprefix{line}{Line}
\mkfancyprefix{def}{Definition}
\mkfancyprefix{itm}{Item}
\mkfancyprefix{app}{Appendix}
\mkfancyprefix{rem}{Remark}
\newcommand{\cref}[1]{\Fref{#1}}

\def\ve#1{{\mathchoice{\mbox{\boldmath$\displaystyle #1$}}%
              {\mbox{\boldmath$\textstyle #1$}}%
              {\mbox{\boldmath$\scriptstyle #1$}}%
              {\mbox{\boldmath$\scriptscriptstyle #1$}}}}

\definecolor{sunset}{rgb}{1,0.5,.05}
\definecolor{marine}{rgb}{0,0,.7}
\definecolor{navy}{rgb}{0,0,.5}
\definecolor{forest}{rgb}{0,.6,0}
\definecolor{brown}{rgb}{0.59, 0.29, 0.0}

\definecolor{darkgreen}{rgb}{0,0.6,0}
\definecolor{darkred}{rgb}{0.7,0,0}

\newcommand{\old}[1]{}

\newcommand\Osoft{O^{\scriptscriptstyle \sim}\!}

\AtBeginDocument{}

\newcommand{\F}{\mathbb{F}}
\newcommand{\Fq}{\mathbb{F}_q}
\newcommand{\ZZ}{\mathbb{Z}}

\renewcommand{\vec}{\ve}

\newcommand{\x}{\ve x}
\renewcommand{\a}{\ve a}

\renewcommand{\b}{\ve b}

\renewcommand{\c}{\ve c}

\newcommand{\wt}{\mathrm{wt}}

\renewcommand{\r}{\ve r}
\renewcommand{\c}{\ve c}

\newcommand{\e}{\ve e}

\newcommand{\Eset}{\mathcal{E}}
\newcommand{\IntDegree}{m}
\newcommand{\MultParam}{s}
\newcommand{\PowParam}{\ell}

\newcommand{\tauKL}{\tau_\mathrm{KL}}

\newcommand{\tauCS}{\tau_\mathrm{CS}}
\newcommand{\tauWZB}{\tau_\mathrm{WZB}}
\newcommand{\tauRP}{\tau_\mathrm{new}}
\newcommand{\RS}{\mathcal{C}_\mathrm{RS}}
\newcommand{\IRS}{\mathcal{C}_\mathrm{IRS}}

\newcommand{\round}[1]{\lfloor #1 \rfloor}

\renewcommand{\j}{\ve j}
\renewcommand{\i}{\ve i}
\newcommand{\f}{\ve f}
\newcommand{\OmegaVec}{\ve \Omega}
\newcommand{\R}{\ve R}

\newcommand{\smaller}{\preceq}
\newcommand{\len}[1]{|#1|}

\newcommand{\Pade}{Pad\'e\xspace}

\newcommand{\chat}{\hat{\c}}
\newcommand{\psihat}{\hat{\psi}}
\newcommand{\lambdahat}{\hat{\lambda}}
\newcommand{\fhat}{\hat{\f}}
\newcommand{\0}{\ve 0}
\newcommand{\Rhat}{\hat{\R}}
\newcommand{\myindex}{\ve h}

\newcommand{\CountLt}[1]{\binom{\IntDegree+#1-1}{\IntDegree}}
\newcommand{\CountLeq}[1]{\binom{\IntDegree+#1}{\IntDegree}}

\newcommand{\WeightedLt}[1]{\IntDegree\binom{\IntDegree+#1-1}{\IntDegree+1}}
\newcommand{\CountLtt}[1]{\tbinom{\IntDegree+#1-1}{\IntDegree}}
\newcommand{\CountLeqt}[1]{\tbinom{\IntDegree+#1}{\IntDegree}}
\newcommand{\WeightedLeqt}[1]{\IntDegree\tbinom{\IntDegree+#1}{\IntDegree+1}}
\newcommand{\WeightedLtt}[1]{\IntDegree\tbinom{\IntDegree+#1-1}{\IntDegree+1}}
\newcommand{\word}[1]{\textnormal{#1}}
\newcommand\modop{\ \word{mod}\ }